\documentclass{article}
\usepackage{graphicx}
\usepackage{fullpage}
\usepackage{enumerate}
\usepackage{amsmath,amsfonts,amsthm}

\graphicspath{{./figs/}}

\usepackage{xspace}
\usepackage{cite}
\usepackage{caption}
\usepackage{subcaption}

\newcommand{\mathset}[1]{\ensuremath {\mathbb {#1}}}

\newcommand{\R}{\mathset{R}}
\newcommand{\etal}{\emph{et~al.}\xspace}

\DeclareMathOperator{\DT}{DT}

\DeclareMathOperator{\RNG}{RNG}
\DeclareMathOperator{\EMST}{EMST}

\theoremstyle{plain}
\newtheorem{theorem}{Theorem}[section]
\newtheorem{lemma}[theorem]{Lemma}
\newtheorem{observation}[theorem]{Observation}

\begin{document}

\title{Time-Space Trade-Offs for Computing Euclidean Minimum 
Spanning Trees\thanks{A preliminary version appeared as 
B.~Banyassady, L.~Barba, and W.~Mulzer. 
\emph{Time-Space Trade-Offs for Computing Euclidean Minimum 
Spanning Trees}. Proc~13th LATIN, pp. 108--119, 2018.
B.B.\@ and W.M.\@ were supported in part by 
DFG project MU/3501/2 and ERC StG 757609.
 L.B.\@ was supported by the ETH Postdoctoral Fellowship.}}

\author{Bahareh Banyassady\footnote{Zuse Institut Berlin, 
    Germany,
  \texttt{bahareh.banyassady@fu-berlin.de}} \and
Luis Barba\footnote{ETH Zurich, Zurich, Switzerland,
\texttt{luis.barba@inf.ethz.ch}} \and
Wolfgang Mulzer\footnote{Institut f\"ur Informatik,
Freie Universit\"at Berrlin, Germany,
\texttt{mulzer@inf.fu-berlin.de}} 
}

\maketitle          

\begin{abstract}
We present time-space trade-offs for computing 
the \emph{Euclidean minimum spanning tree} 
of a set $S$ of $n$ point-sites in the 
plane. More precisely, we assume that $S$ 
resides in a random-access memory that can 
only be read. The edges of the Euclidean minimum 
spanning tree $\EMST(S)$ have 
to be reported sequentially, and they cannot 
be accessed or modified afterwards. There is 
a parameter $s \in \{1, \dots, n\}$ so that 
the algorithm may use $O(s)$ cells of 
read-write memory (called the \emph{workspace}) 
for its computations. Our goal is to 
find an algorithm that has the best 
possible running time for any given $s$ 
between $1$ and $n$. 

We show how to compute $\EMST(S)$ in 
$O\big((n^3/s^2)\log s \big)$ time with $O(s)$ 
cells of workspace, giving a smooth trade-off 
between the two 
best known bounds $O(n^3)$ for $s = 1$ and 
$O(n \log n)$ for $s = n$. For this, we run 
Kruskal's algorithm on the \emph{relative 
neighborhood graph} (RNG) of $S$. It
is a classic fact that the minimum spanning 
tree of $\RNG(S)$ is exactly $\EMST(S)$.
To implement Kruskal's algorithm with $O(s)$ 
cells of workspace, we define \emph{$s$-nets},
a compact representation of planar graphs.
This allows us to efficiently maintain and 
update the components of the current 
minimum spanning forest as the edges are being
inserted.

%\textbf{Keywords:} Euclidean minimum spanning tree, 
%relative neighborhood graph, time-space trade-off, 
%limited workspace model, Kruskal's algorithm
\end{abstract}

\section{Introduction}

Given a set $S$ of $n$ point-sites in the plane, the 
\emph{Euclidean minimum spanning tree} 
of $S$, $\EMST(S)$, is 
the minimum spanning tree 
of the complete graph with vertex set $S$, where 
the weight of an edge between two point-sites is the 
Euclidean distance between them. 
The problem of computing $\EMST(S)$ efficiently 
constitutes a core question of computational geometry, 
and it is discussed in virtually every introductory 
course on the subject. There are several algorithms 
that find $\EMST(S)$ in $O(n \log n)$ time
and with $O(n)$ cells of space~\cite{deBergChvKrOv08,PreparataSh85}. 

Here, our goal is to design algorithms to compute 
$\EMST(S)$ in the \emph{limited-workspace model},
where only a limited number of memory cells are 
available for reading and writing during the 
execution of the algorithm~\cite{BanyassadyKoMu18}.
This model is of interest theoretically because it 
provides a trade-off between 
the running time and the space usage of an algorithm. 
It is also useful from a practical point of view, in developing software 
for portable devices and sensors where memory is
the limiting factor. A significant amount 
of research has focused on the design of 
algorithms under memory constraints. Much of this 
work dates from the 1970s, when memory was an 
expensive commodity. Even today, while this cost has dropped 
substantially, at the same time 
the amount of data has increased, and the size 
of some devices has been reduced dramatically. 
In particular, sensors and small devices, where 
larger memories are neither possible nor desirable, 
have proliferated in recent years. Moreover, even 
if a device is equipped with a large memory, 
it may still be preferable to limit the number 
of write operations. For example, writing to 
flash memory is slow, and it may 
reduce the lifetime of the memory. Additionally, 
if the input is stored on removable devices, 
write-access may not be allowed due to technical 
or security concerns.

There are many variants of the 
limited-workspace model~\cite{BanyassadyKoMu18}, but the general
outline is usually the same: the input resides 
in a read-only memory and cannot 
be modified directly by the algorithm. Instead, 
the algorithm may use a controlled amount of 
storage cells (usually called \emph{workspace})
that reside in a local memory and can be 
modified as needed to solve the problem.
Since the result of the computation may not fit 
in the local memory, the model provides
a write-only memory where the output is 
reported sequentially.
One noteworthy instance of the model is encountered in 
computational complexity theory, where the 
complexity class LOGSPACE consists of all 
decision problems that can be solved with 
a deterministic Turing machine that has access 
to two tapes~\cite{AroraBa09}. The first tape is read-only and 
contains the input, while the second tape represents 
the workspace and contains a logarithmic 
(in the input size) number of read-write bits. 
In other words, the second tape stores only a 
constant number of \emph{words} with a logarithmic
number of bits that can be used as counters or
as pointers to the input. Thus, the computational
model represented by LOGSPACE  is sometimes 
referred to as the \emph{constant-workspace} 
model~\cite{AsanoMuWa11,AsanoMuRoWa11}. 

More generally, we may allow the algorithm 
to use a workspace of $O(s)$ cells, 
for some parameter $s$, where a cell stores either 
an input item (such as a point coordinate), a pointer 
into the input structure (of logarithmic size in the 
input length), or a counter (with a logarithmic
number of bits). 
The goal is to design algorithms 
whose running time decreases as $s$ increases, 
and to obtain a smooth trade-off between 
workspace size and running time. 

\paragraph{Our results.} 
For computing the Euclidean minimum spanning tree
of $n$ given point-sites in the plane in the constant-workspace model, 
Asano~\etal~\cite{AsanoMuRoWa11} 
presented an algorithm that runs in  $O(n^3)$ time.  
We use their method as a starting 
point for a time-space trade-off. 
As a result, we obtain an algorithm that, for any given number 
$s \in \{1, \dots, n\}$ of workspace cells, 
computes the EMST in $O\big((n^3/s^2)\log s\big)$ 
time. This yields a smooth transition between the 
$O(n^3)$ time algorithm for $s = 1$ by
Asano~\etal~\cite{AsanoMuRoWa11} 
and the classic $O(n \log n)$ algorithm for
$s = n$~\cite{deBergChvKrOv08,PreparataSh85}.

As a main tool,
we define a compact representation of a
plane graph $G$, called the \emph{$s$-net}. The 
$s$-net consists of a ``dense'' set of 
$s$ edges in $G$ for which we remember the 
edge-face incidences. That is, for each edge $e$
in the $s$-net, we store the (at most two) faces of $G$
to which $e$ is incident. Furthermore, for each
face in $G$ that has at least one incident edge
in the $s$-net, we store the order in which the incident 
edges of the $s$-net appear. The density property 
guarantees that we cannot walk  
for more than $O(s)$ steps along a connected component of the boundary of a
face in $G$ without reaching an 
edge of the $s$-net. 
This turns out to be useful for an efficient limited-workspace
implementation of Kruskal's
MST-algorithm on a plane graph $G$.
Recall that in this algorithm, the edges of $G$ are inserted 
into an auxiliary graph by increasing order of weight. 
To insert a new edge $e$, we need to determine whether 
the endpoints of $e$ are in the same component 
of the current auxiliary graph. If $G$ is plane, 
this amounts to testing whether the
endpoints of $e$ are incident to the same face of the 
current auxiliary graph---precisely the task for which 
$s$-nets were created.
While the $s$-net is designed to speed up Kruskal's 
algorithm, this structure may be of independent interest, 
as it provides a compact way to represent plane graphs that 
may be useful in other problems.

\paragraph{Related work.} 
The study of constant-workspace algorithms in theoretical
computer science started with the
complexity class LOGSPACE~\cite{AroraBa09}. Since then, 
many classic problems were considered in this setting.
For example, there are a lot of relevant results on 
selection and sorting~\cite{MunroPa80,MunroRa96,PagterRa98,ChanMuRa14}. 
A long-standing algorithmic problem in graph
theory was eventually solved by Reingold~\cite{Reingold08},
who showed that the reachability between two 
vertices in an undirected graph can be decided in LOGSPACE. 
The model was made popular in computational geometry 
by Asano~\etal~\cite{AsanoMuRoWa11}, 
who presented several algorithms to compute classic
geometric structures in the constant-workspace 
model (see the recent survey~\cite{BanyassadyKoMu18}). 
Time-space trade-offs for many of these structures 
were presented in subsequent years~\cite{AsanoKi13,
BarbaKoLaSaSi15,KormanMuvReRoSeSt17,BanyassadyKoMuvReRoSeSt18, 
BahooBaBoDuMu19,AronovKoPrvReRo16,BarbaKoLaSi14,
DarwishEl14,HarPeled16,AhnBaOhSi17,ElmasryKa16}, with the notable exception 
of the EMST. This is  finally addressed here.

\section{Preliminaries and Notation}

We recall the basic definition and some 
properties of the Euclidean minimum
spanning tree, and we briefly review 
some known algorithms for computing it,
both in the classic setting and in the 
constant-workspace model.
Furthermore, we recall the definition 
of the \emph{relative neighborhood graph} (RNG), 
a basic proximity structure defined on planar point sets,
and we discuss the relationship between RNGs and Euclidean 
minimum spanning trees.

\subsection{Euclidean Minimum Spanning Trees}\label{sec:def-emst}

Let $S = \{p_1, \dots, p_n\} \subset \R^2$ be 
a set of $n$ point-sites in the plane, from now on referred to 
as sites. 
We assume that $S$ is in \emph{general position}, 
i.e., no three sites lie on a common line, 
no four sites lie on a common circle, and the 
pairwise distances between the sites are
all distinct. Let $G_S$ be the complete weighted 
graph with vertex set $S$, where the edges are 
weighted with the Euclidean distance between 
their endpoints. A minimum spanning
tree of $G_S$ is called a \emph{Euclidean minimum
spanning tree} of $S$, and it is denoted by $\EMST(S)$,
see Figure~\ref{fig:trees}. 

\begin{figure}[t]
	\centering
	\subcaptionbox{\label{fig:spanning-tree}}
	{\includegraphics[page=1]{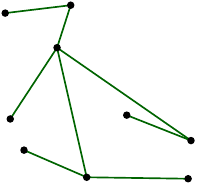}}
	\hspace{1in}
	\subcaptionbox{\label{fig:emst}}{\includegraphics[page=2]{Spanning-Tree}}
	\caption{A set $S$ of sites in the plane, and (a) a spanning tree 
		of the complete graph with vertex set $S$ and (b) the Euclidean 
		minimum spanning tree of $S$.}
	\label{fig:trees}
\end{figure}

Under our general position assumption, it is known that $\EMST(S)$ is unique
(see, e.g.,~\cite{CormenLeRiSt09}). 
Given $S$, we would like to report the edges of $\EMST(S)$
in any order, so that each edge 
is listed exactly once.

\paragraph{A Classic Algorithm.}
We recall the classic algorithm by Kruskal~\cite{CormenLeRiSt09}:
we start with an empty forest $F$, and we consider 
the edges of $G_S$ one by one, by increasing weight. 
In each step, we insert the current edge $e$ into $F$ if and 
only if there is no path in $F$ between the endpoints of $e$;
see Figure~\ref{fig:emst-e1-e2}. 
After all edges of $G_S$ have been considered,
the final graph $F$ is exactly $\EMST(S)$. 

\begin{figure}[ht]
	\centering
	\subcaptionbox{}{\includegraphics[page=3]{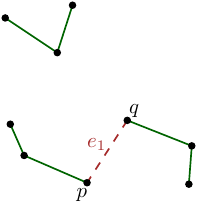}}\hspace{1in}
	\subcaptionbox{}{\includegraphics[page=4]{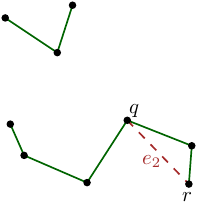}}
	\caption{Two steps in Kruskal's algorithm for computing
		the EMST for a set of sites in the plane.
		(a) In this step, the algorithm adds $e_1$ to $F$, 
		since there is no path between $p$ and $q$ in $F$.
		(b) In the next step, $e_2$ is discarded,
		since $q$ and $r$ lie in the same component of $F$.}
	\label{fig:emst-e1-e2}
\end{figure}

During the above procedure, using a \emph{disjoint set-union
structure}, we keep track of the components
of $F$ so that we can determine if 
there is a path in $F$ between the two endpoints of
the next edge $e$~\cite{CormenLeRiSt09}.
With an efficient implementation of the disjoint set-union
structure, the time for inserting the edges into $F$ is dominated by
the time for sorting the edges with their weight.
This gives a running time of $O(n^2\log n)$ with $O(n)$
cells of workspace. 

The running time can be improved as follows:
the \emph{Delaunay triangulation} of $S$, $\DT(S)$,
is the triangulation of $S$ in which three sites
$p, q, r$ form a triangle if and only if the
disk with $p$, $q$, and $r$ on the boundary contains
no other sites from $S$ in its interior~\cite{deBergChvKrOv08}.
Under our general position 
assumption, this defines a unique plane triangulation
of $S$ which is a supergraph of $\EMST(S)$; see 
Figure~\ref{fig:dt-emst}~\cite{deBergChvKrOv08}.
Thus, $\EMST(S)$ is the minimum spanning tree of $\DT(S)$,
and it suffices to consider the $O(n)$ edges of $\DT(S)$ instead
of the $O(n^2)$ edges of $G_S$. 
Then, Kruskal's algorithm 
runs in $O(n\log n)$ time, when $O(n)$ cells of 
workspace are available~\cite{deBergChvKrOv08}.

\begin{figure}[ht]
	\centering
	\includegraphics[page=5]{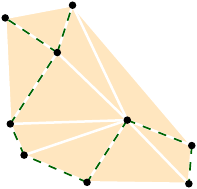}
	\caption{The Delaunay triangulation $\DT(S)$ and the 
	Euclidean minimum spanning tree $\EMST(S)$ for a
	planar point set $S$. The dashed edges belong to 
	$\EMST(S)$.}\label{fig:dt-emst}
\end{figure}

\paragraph{The Constant-Workspace Algorithm.}
Asano~\etal~\cite{AsanoMuRoWa11} presented an algorithm
that reports the edges of $\EMST(S)$ in $O(n^3)$ time with 
$O(1)$ cells of workspace. 
Like the classic method, their algorithm uses the fact that
$\EMST(S)$ is a subgraph of $\DT(S)$. 
First, Asano~\etal~show that there 
exists a constant-workspace algorithm 
that solves the following task in $O(n)$ time: 
given an edge $pq$ of $\DT(S)$, find the next 
edge $pr$ of $\DT(S)$ that is incident to $p$ 
after $pq$ in clockwise direction.
Using the fact that for each $p \in S$,
the edge between $p$ and its nearest neighbor
in $S \setminus \{p\}$ belongs to $\DT(S)$,
this gives an algorithm
that reports the edges of 
$\DT(S)$, one by one, in an arbitrary order, 
in $O(n)$ time per edge. We will not
describe the details 
here, but we present an analogous result for relative 
neighborhood graphs in 
Section~\ref{sec:computing-rng}.

Then, the algorithm of Asano~\etal~to
list the edges of $\EMST(S)$ proceeds as follows:
we run the constant-workspace algorithm
that enumerates the edges of $\DT(S)$.
Every time a new edge $e$ of $\DT(S)$
is reported, 
we test if $e$ is in $\EMST(S)$. If so,
we output $e$; otherwise, we discard it.
To perform this test, we consider the
subgraph $\DT_{<e}$ of $\DT(S)$
that contains all the edges of length 
less than $|e|$, where $|e|$ denotes the
(Euclidean) length of $e$. 
By the cut-property of minimum spanning trees,
it follows that $e$ is \emph{not} in $\EMST(S)$
if and only if 
the endpoints $p$ and $q$ of $e$ lie 
in the same connected component of 
$\DT_{<e}$. 
Since $\DT(S)$ is plane, this means that
$e$ is not in $\EMST(S)$ 
if and only if $p$
and $q$ lie on a common connected component of 
the boundary of the face of $\DT_{<e}$ that contains
$e$.
In other words, $e \notin \EMST(S)$ if and only if we 
encounter $q$ by walking 
from $p$ along the connected component 
of the boundary of the face of $\DT_{<e}$ that contains $e$;
see Figure~\ref{fig:dt<e}.

\begin{figure}[ht]
\centering
\includegraphics{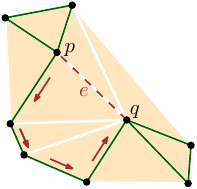}
\caption{The subgraph $\DT_{<e}$ for a planar point set $S$ 
and an edge $e=pq$ of $\DT(S)$. To decide if $e$ belongs to $\EMST(S)$,
we check if $p$ to $q$ are in the same connected component of $\DT_{<e}$.
For this, we walk along the boundary of the face of $\DT_{<e}$ that 
contains $e$, starting from $p$.}\label{fig:dt<e}
\end{figure}

To perform one step of this walk,  
we use the above-mentioned subroutine due to Asano~\etal that receives an
edge of $\DT(S)$ and finds the next clockwise edge of
$\DT(S)$ 
using $O(n)$ time and $O(1)$ cells
of workspace. 
We start with the edge $e = pq$ and we repeatedly call
the subroutine until we find the first
Delaunay edge $pr$ of length less than $|e|$
(if no such edge exists, then $e$ belongs to $\EMST(S)$). 
Then, we repeatedly call the subroutine, starting
with the reverse edge $rp$, until we encounter an
edge $rs$ with length less than $|e|$. We continue until
either (i) we encounter $q$, in which case $e$ does
not belong to $\EMST(S)$; or (ii) the subroutine 
produces the edge $e$, which means that we have
traversed the complete connected component of the 
face boundary without seeing
$q$, in which case $e$ belongs to $\EMST(S)$.\footnote{Note that
it does not suffice to stop the walk once we come back to $p$,
because several edges that are incident to $p$ might appear on the
boundary of the relevant face.}
During this walk, each edge of $\DT(S)$ 
is generated at most twice, at most once for each endpoint. 
Thus, we need $O(n^2)$ time to decide if 
an edge $e$ of $\DT(S)$ is in $\EMST(S)$, with
$O(1)$ cells of workspace.

Since $\DT(S)$ has $O(n)$ edges, and since it takes
$O(n^2)$ time to decide membership in 
$\EMST(S)$, the total time to find all the edges of
$\EMST(S)$ is $O(n^3)$. 
The overhead for computing the edges of
$\DT(S)$ in the outer loop is $O(n^2)$, which is 
negligible compared to the remainder of the algorithm. 
The workspace is constant.  
We can also report
the edges of $\EMST(S)$ by increasing length: we repeatedly list 
all edges of $\DT(S)$, and each time we find the shortest edge 
$e \in \DT(S)$ whose membership in $\EMST(S)$ has not yet been checked, we
apply our test to $e$. 
Now, the overhead for the outer loop is $O(n^3)$ instead of
$O(n^2)$, 
without any effect on
the total asymptotic running time.

\subsection{Relative Neighborhood Graphs}\label{sec:def-rng}

The \emph{relative neighborhood graph} is a geometric
structure that 
``lies between'' the Euclidean minimum spanning tree
and the Delaunay triangulation.
For two sites $p,q \in S$, we define the \emph{lens} of $p$ and $q$
as the intersection of the disk centered at $p$ with radius $|pq|$
and the disk centered at $q$ with radius $|pq|$,
where $|\cdot|$ denotes the Euclidean distance.
The lens of $p$ and $q$ is called \emph{empty} if 
it contains no sites of $S \setminus \{p, q\}$ in
its interior. In other words, 
the two sites $p$ and $q$ have the \emph{empty lens} property if
there is no site $r \in S \setminus \{p, q\}$ such that both 
$|pr|$ and $|qr|$ are shorter than $|pq|$;
see Figure~\ref{fig:lens}.

\begin{figure}[ht]
\centering
\includegraphics[page=1]{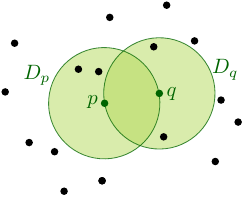}
\caption{A set $S$ of sites and two sites $p, q \in S$. 
The disks $D_p$ and $D_q$ have radius $|pq|$ and are centered at 
$p$ and $q$, respectively. The two sites $p$ and $q$ satisfy the
empty lens property since $D_p \cap D_q$ is empty of 
other sites of $S$.}\label{fig:lens}
\end{figure}

The \emph{relative neighborhood graph} $\RNG(S)$ of $S$ is the undirected 
graph with vertex set $S$ obtained by connecting two sites 
$p, q \in S$ with an edge if and only if the lens of $p$ and $q$ 
is empty~\cite{Toussaint80}. 
One can show that a plane embedding of $\RNG(S)$
is obtained by drawing the edges as straight line
segments between the corresponding sites in $S$; 
see Figure~\ref{fig:emst-RNG(S)}. 
By definition, $\RNG(S)$ is a subgraph of 
$\DT(S)$.\footnote{If an edge $e = pq$ is 
in $\RNG(S)$, then the lens of $p$ and $q$ is 
empty, which also means that the smallest 
disk with both $p$ and $q$ on the boundary is empty 
of other sites of $S$. Thus, $e$ belongs to $\DT(S)$.}
Furthermore, it is well-known that $\EMST(S)$ is a subgraph of  
$\RNG(S)$~\cite{deBergChvKrOv08}. In particular, this 
implies that $\RNG(S)$ is connected; see Figure~\ref{fig:emst-RNG-DT}. 
Each vertex in $\RNG(S)$ has at most 
six neighbors, so $\RNG(S)$ has bounded
degree and $O(n)$ edges.
We will denote the number of those edges by $m$.
Given $S$, we can list the edges of $\RNG(S)$
in $O(n \log n)$ time using $O(n)$ cells of
workspace~\cite{Toussaint80,JaromczykTo92,MitchellMu18}.

\begin{figure}[t]
\centering
\includegraphics[page=2]{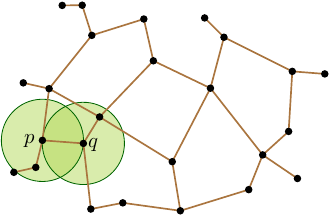}
\caption{An edge $pq$ in $\RNG(S)$.
The lens of $p$ and $q$ is empty.}\label{fig:emst-RNG(S)}
\end{figure}

\begin{figure}[ht]
	\centering
	\includegraphics[page=3]{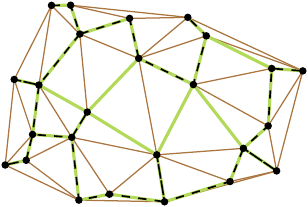}
	\caption{An illustration of the fact 
		$\EMST(S) \subseteq \RNG(S) \subseteq \DT(S)$.
		The dashed black edges belong to $\EMST(S)$ and are
		a subset of the green edges which represent $\RNG(S)$.
		All these edges form a subset of the edges of the underlying graph $\DT(S)$.}
	\label{fig:emst-RNG-DT}
\end{figure}

Thus, we can compute 
$\EMST(S)$ with the algorithm of Kruskal using 
the edges of $\RNG(S)$ instead of $\DT(S)$.
Since both $\RNG(S)$ and $\DT(S)$ have $O(n)$ edges,
this does not improve the running time of Kruskal's algorithm 
in the classic setting where $O(n)$ cells of workspace are available.
However, since $\RNG(S)$ (unlike $\DT(S)$) has bounded 
degree, it turns out to be the superior choice for the 
limited-workspace model.

\begin{figure}[ht]
	\centering
	\includegraphics{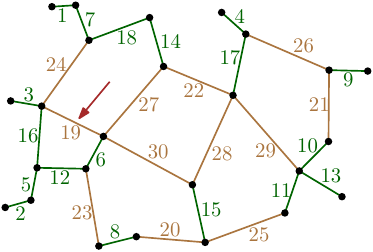}
	\caption{The $\RNG$ for a set $S$ of sites. The labels represent the indices of
		the edges in the sorted sequence $E_R$. The subgraph $\RNG_{19}$ is shown
		in green. The edge $e_{19}$ does not belong to $\EMST(S)$ since its endpoints 
		lie in the same component of $\RNG_{19}$.}
	\label{fig:kruskal-rng}
\end{figure}

We define $E_R = e_1, \dots, e_m$ to be the sorted 
sequence of edges in $\RNG(S)$, in increasing order
of length. For $i \in \{1, \dots, m\}$, we define 
$\RNG_i$ to be the subgraph of $\RNG(S)$ with 
vertex set $S$ and edge set $\{e_1, \dots, e_{i-1}\}$.
Thus, to check if $e_i$ belongs to $\EMST(S)$, the algorithm
by Kruskal checks if the endpoints of $e_i$ lie on the
same component of $\RNG_i$; see Figure~\ref{fig:kruskal-rng}.

In our algorithms, we represent each edge $e_i \in E_R$ by 
two directed \emph{half-edges}. The two
half-edges are oriented in opposite directions such that the 
face incident to a half-edge lies on its left. We call the endpoints of 
a half-edge the \emph{head} and the \emph{tail} such that the
half-edge is directed from the tail endpoint to the head endpoint.
Furthermore, directed half-edges will be denoted as
$\overrightarrow{e}$ and undirected edges as $e$;
see Figure~\ref{fig:head-tail} for an illustration.

\begin{figure}[ht]
	\centering
	\includegraphics{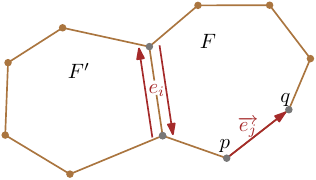}
	\caption{A schematic drawing of the faces $F, F'$ of $\RNG(S)$. 
		The two half-edges that correspond to the edge $e_i$ are 
		oriented such that the face
		incident to each lies on its respective left.
		The sites $p$ and $q$ are the head and the tail 
		endpoints of the half-edge 
		$\protect\overrightarrow{e_j} = \protect\overrightarrow{pq}$, respectively.}
	\label{fig:head-tail}
\end{figure}

Using the concept of half-edges, we define the \emph{face-cycle} in a 
planar graph. For $i \in \{1, \dots, m\}$, 
a \emph{face-cycle} in $\RNG_i$ is the circular sequence of consecutive
half-edges such that (i) they bound either a face in $\RNG_i$ or the 
outer face 
in a connected component of $\RNG_i$\footnote{Since $\RNG_i$ has several
connected components, to define face-cycles of the outer face,
we have to consider the outer face of each connected component 
individually.}; and
(ii) every two consecutive half-edges $e$ and $e'$ in a face-cycle
share an endpoint which is the head
vertex of $e$ and the tail of $e'$.

The definition implies that all the
half-edges in a face-cycle are oriented in the same direction and
the face (or the outer face) incident to the half-edges lies on their left.
Note that every half-edge lies on only one face-cycle; however, a
site of $S$ might be on several face-cycles; see Figure~\ref{fig:emst-face-cycle}.
The \emph{partial relative neighborhood graph} $\RNG_i$ can be represented as a collection of face-cycles.

\begin{figure}[ht]
	\centering
	\includegraphics{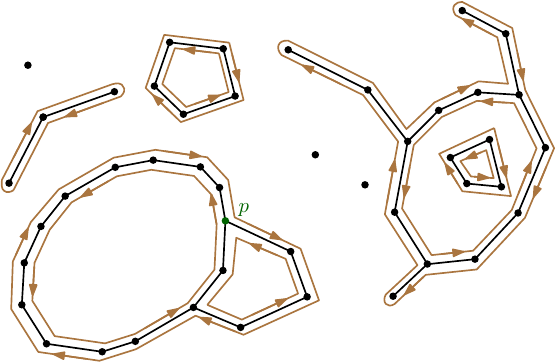}
	\caption{A schematic drawing of $\RNG_i$ for a planar set $S$ of sites.
		The edges are shown in black. The face-cycles of $\RNG_i$ are in beige. 
		The half-edges of each face-cycle are directed
		according to the arrows on the corresponding cycle. 
		The site $p \in S$ is on three face-cycles of $\RNG_i$. 
		Each of the six half-edges
		incident to $p$ is only on one face-cycle.}
	\label{fig:emst-face-cycle}
\end{figure}

Let $j\geq i\geq 1$. For a half-edge $\overrightarrow{e_j}$ with head $q$, we define the \emph{predecessor} and the \emph{successor}
of $\overrightarrow{e_j}$ in $\RNG_i$ as follows:
the predecessor $\text{pre}(\overrightarrow{e_j})$ of 
$\overrightarrow{e_j}$ is the half-edge in $\RNG_i$ which has
$q$ as its head and is the first half-edge encountered in a 
counterclockwise sweep from $\overrightarrow{e_j}$ around $q$.
The successor $\text{suc}(\overrightarrow{s_j})$ of $\overrightarrow{e_j}$
is the half-edge in $\RNG_i$ which has
$q$ as its tail and is the first half-edge encountered in a 
clockwise sweep from $\overrightarrow{e_j}$ around $q$;
see Figure~\ref{fig:emst-predecessor-successor} for an illustration. 
Note that, if there is no edge incident to $q$ in $\RNG_i$, we set both
the predecessor and the successor to \emph{Null}.

Let $i > j\geq 1$. For the half-edge $\overrightarrow{e_j}$ in $\RNG_i$
that lies on a face-cycle $F$, we define the \emph{next} edge of
$\overrightarrow{e_j}$ on $F$ as the half-edge
on $F$ whose tail is the head of $\overrightarrow{e_j}$.
Note that
the next edge of a half-edge $\overrightarrow{e_j}$ is defined with 
respect to each diagram $\RNG_i$ with $i > j$ and thus 
$\overrightarrow{e_j} \in \RNG_i$, whereas
the predecessor and successor of $\overrightarrow{e_j}$ 
are defined with respect to each diagram $\RNG_i$ with $i \leq j$, 
meaning that $\overrightarrow{e_j} \not \in \RNG_i$.

\begin{figure}[ht]
	\centering
	\includegraphics[page=2]{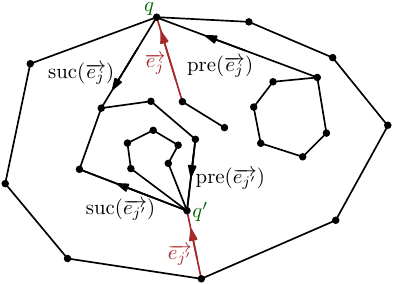}
	\caption{A schematic drawing of $\RNG_i$ and  
             the half-edges $\protect\overrightarrow{e_j}$ with 
             head $q$ and $\protect\overrightarrow{e_{j'}}$ with 
             head $q'$, for $j, j' \geq i \geq 1$. 
             The predecessor and successor of
             $\protect\overrightarrow{e_j}$
             are $\text{pre}(\protect\overrightarrow{e_j})$ and 
             $\text{suc}(\protect\overrightarrow{e_j})$, respectively.
             The predecessor and successor of
             $\protect\overrightarrow{e_{j'}}$             
             are $\text{pre}(\protect\overrightarrow{e_{j'}})$ and 
            $\text{suc}(\protect\overrightarrow{e_{j'}})$, respectively.}\label{fig:emst-predecessor-successor}
\end{figure}

\section{Computing the Relative Neighborhood Graph}
\label{sec:computing-rng}

For the given set $S = \{p_1, \dots, p_n \}$ of sites, 
our first goal is to compute the edges of $\RNG(S)$ in the 
limited-workspace model.
We first present an algorithm for
listing the edges of $\RNG(S)$ in an arbitrary order,
using $O(s)$ cells of workspace.
Then, we extend the algorithm so that it outputs
the edges in sorted order according to their lengths. 
Our method is inspired by the time-space trade-off 
for Voronoi diagrams by 
Banyassady~\emph{et al.}~\cite{BanyassadyKoMuvReRoSeSt18}.

\subsection{All the Incident Edges to Some Sites}

The idea is to subdivide $S$ into \emph{batches} 
of $s$ sites, and to compute all the edges incident to the sites 
in one batch simultaneously.
In the following lemma, we explain how to process one batch
using $O(s)$ cells of workspace. 
This lemma is the main reason why we prefer to use $\RNG(S)$ 
instead of $\DT(S)$, the choice of Asano 
\emph{et al.}~\cite{AsanoMuRoWa11}.
More precisely, in $\DT(S)$, there may be sites of high degree, so 
that we cannot guarantee that all edges incident to the sites of 
a single batch can be found in the desired time.

\begin{lemma}\label{lem:6neighbors}
	Let $S$ be a planar set of $n$ point-sites in general position,
	stored in a read-only array. Given a set $Q \subseteq S$ of  
	$s$ sites, we can compute for each $p \in Q$ 
	the neighbors of $p$ in $\RNG(S)$ (for each $p$, there
	are at most six neighbors)
	in total time $O(n \log s)$ and using
	$O(s)$ cells of workspace.
\end{lemma}

\begin{proof}
	The algorithm has two phases. In the first phase, for each 
	$p \in Q$, we find a set containing the neighbors of $p$ 
    in $\RNG(S)$. This superset has size at most six.
	In the second phase, we check for each $p \in Q$ 
    which of these \emph{candidate neighbors}
	are the actual neighbors of $p$ in $\RNG(S)$.
	
	The first phase proceeds in $\lceil n/s \rceil$ \emph{steps}.
	In each step, we process a \emph{batch} of $s$ sites of 
	$S = R_1\cup \dots \cup R_{\lceil n/s\rceil}$, and we produce 
    at most six candidate neighbors
	for each $p \in Q$. In the first step, we take the 
	first batch $R_1 \subseteq S$ of $s$ sites, and we compute 
    $\RNG(Q \cup R_1)$. 
	Because $|Q \cup R_1| \leq 2s$, we can do this in $O(s \log s)$ 
    time using known algorithms~\cite{Toussaint80,JaromczykTo92,MitchellMu18}. 
    For each $p \in Q$, we remember 
	the neighbors of $p$ in $\RNG(Q \cup R_1)$
	(there are at most six neighbors). Notice 
    that if for a pair 
	$p \in Q, r \in R_1$, the edge $pr$ is not in 
    $\RNG(Q \cup R_1)$, then the lens of 
	$p$ and $r$ is non-empty. 
	This also means that $pr$ is not an edge of $\RNG(S)$. 
	Let $N_1$ be the set containing all neighbors in $\RNG(Q \cup R_1)$ 
    of all sites in $Q$.
	Storing $N_1$, the set of candidate neighbors, 
    requires $O(s)$ cells of workspace. 
	
	Then, in each step $j = 2, \dots, O(n/s)$, we take the next batch $R_j \subseteq S$
	of $s$ sites, and we compute $\RNG(Q \cup R_j\cup N_{j-1})$ in $O(s\log s)$ time 
	using $O(s)$ cells of workspace. For each $p\in Q$, we store the set 
	of 
	neighbors of $p$ in this computed graph (this set has size at most
	six). Additionally, we let $N_j$ be the set 
	containing all neighbors in $\RNG(Q \cup R_j\cup N_{j-1})$ of all sites in $Q$.
	Note that $N_j$, the set of candidate neighbors, consists of $O(s)$ sites as 
	each site in $Q$ has a degree of at most six in the computed graph. At this step, 
	we do not need to store $N_{j-1}$ anymore.
	
	After $\lceil n/s \rceil$ steps we are left with at most six 
	candidate neighbors for each site in $Q$. As mentioned above, for a pair $p\in Q, r\in S$, 
	if $r$ is not among the candidate neighbors of $p$, then, at some point in the construction, 
	there was an \emph{obstructing} site inside the lens of $p$ and $r$.
	Therefore, only the candidate neighbors can define edges of 
	$\RNG(S)$, but not necessarily all of them.
	See Figure~\ref{fig:extra-neighbors} for an example.
	
	\begin{figure}[ht]
		\centering
		\includegraphics[page=1]{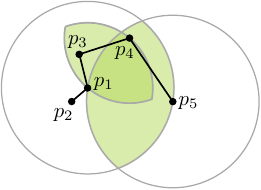}
		\caption{For $S=\{p_1, \dots, p_5\}$, 
			the set of neighbors of $p_1$ in $\RNG(S)$ is $\{p_2, p_3\}$. 
			Suppose that $p_3$ and $p_4$ are processed in some steps before $p_5$.
			After processing $p_3$ and $p_4$, the site $p_4$ is not 
            a candidate neighbor of $p_1$,
			because $p_3$ lies in their lens.
			This results in $p_5$ becoming a candidate neighbor of $p_1$ in one
			of the following steps. Since $p_4$ is the only site in the
			lens of $p_1$ and $p_5$, the site $p_5$ will  remain 
            as a candidate neighbor of $p_1$.
		}
		\label{fig:extra-neighbors}
	\end{figure}
	
	In the second phase, to obtain the edges of $\RNG(S)$ incident 
	to the sites in $Q$, we go again through the entire set 
	$S = R_1\cup \dots \cup R_{\lceil n/s\rceil}$ in batches of 
	size $s$: in the first step, we start with all the sites in $Q$ 
	and their candidate neighbors in $N_{\lceil n/s \rceil}$, and we 
	construct $\RNG(Q \cup R_1 \cup N_{\lceil n/s \rceil})$.
	For each $p \in Q$ and for each candidate neighbor $r$ of $p$ 
	in $N_{\lceil n/s \rceil}$, we check if $r$ is still 
	a neighbor of $p$ in this computed graph. 
	If not, we remove $r$ from the candidate neighbors of $p$.
	We denote the pruned set of candidate neighbors of all the
	sites in $Q$ by $N'_1$. The candidate neighbors in 
	$N_{\lceil n/s \rceil}$ for which there is an obstructing site in $R_1$
	will not appear in $N'_1$.
	
	Then, in each step $j= 2, \dots, O(n/s)$, 
	we construct the graph $\RNG(Q \cup R_1 \cup N'_{j-1})$.
	Again, for each site $p \in Q$, we remove its candidate neighbors in
    $N'_{j-1}$ that are no longer neighbors of $p$ in the computed graph.
	We denote 
	the pruned set of candidate neighbors of all the sites in $Q$ by $N'_{j}$.
	In this step, we do not need to store $N'_{j-1}$ anymore.
	After going through all the batches, the candidates that have survived 
	define the edges of $\RNG(S)$ incident to the sites in $Q$; see
	Figure~\ref{fig:s-sites-rng}. Note
	that in all the steps, $N'_j$ contains at most six candidate neighbors
    for 
	each site of $Q$, and thus, its size is $O(s)$.
	
	\begin{figure}[ht]
		\centering
		\includegraphics[page=2]{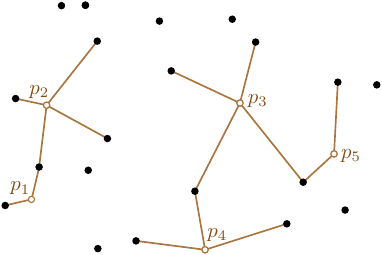}
		\caption{For a set of sites $S$ and $Q=\{p_1, \dots, p_5 \}$, the neighbors
			in $\RNG(S)$ of all the sites in $Q$ are generated.}
		\label{fig:s-sites-rng}
	\end{figure}
	
	Since the algorithm takes $O(s\log s)$ time per step, and since the
	number of steps is $2 \cdot \lceil n/s\rceil$, the total
	running time of the algorithm is $O(n\log s)$.
	The space requirement for storing the candidate neighbors as well as
	the intermediate $\RNG$s is $O(s)$ cells of workspace.
\end{proof}

\subsection{Finding All the Edges of RNG}

Through repeated application of Lemma~\ref{lem:6neighbors},
we can compute all the edges of $\RNG(S)$,
in some arbitrary order, using a workspace of
$O(s)$ cells.

\begin{theorem}\label{thm:rng}
	Suppose we are given a set of $n$ point-sites $S=\{p_1, \dots, p_n \}$
	in the plane
	in general position, stored in a read-only array. Let $s$ be a parameter 
	in $\{ 1, \dots, n \}$. We can compute the edges of $\RNG(S)$ in total time 
	$O \big( (n^2/s) \log s \big)$, using $O(s)$ cells of workspace.
\end{theorem}

\begin{proof}
	We take the set $Q$ of the first $s$ sites of $S$, and we apply 
	Lemma~\ref{lem:6neighbors} on $Q$ to find all the neighbors in $\RNG(S)$ 
	of all the sites in $Q$. Whenever we find a neighbor $p_j$ of a site $p_i$ 
	in $\RNG(S)$, we report the edge $p_ip_j$ only if $i<j$. This guarantees that 
	the edge $p_ip_j$ of $\RNG(S)$ is reported only once. 
	Then,  we take the next batch of $s$ sites of $S$ and
	repeat the same procedure. We continue until all the sites
	in $S$ are processed, i.e., $O(n/s)$ times; see Figure~\ref{fig:comp-rng}.
	
	Lemma~\ref{lem:6neighbors} guarantees that all the reported edges belong to 
	$\RNG(S)$ and all the edges of $\RNG(S)$ are reported exactly once. 
	Regarding the running time of the algorithm, $O(n/s)$ invocations of 
	Lemma~\ref{lem:6neighbors} take a total of $O \big( (n^2/s) \log s \big)$ time.
	The space requirement is immediate.
\end{proof}

\begin{figure}[ht]
	\centering
	\includegraphics[page=3]{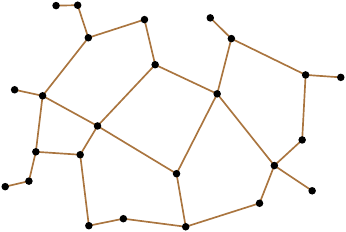}
	\caption{The RNG for the set of sites $S$ is generated
		by processing sites of $S$ in batches of $s$ sites.}
	\label{fig:comp-rng}
\end{figure}

\subsection{Edges of RNG in Sorted Order of Length}

In the following lemma, we use a technique that
is taken from the work of Chan and Chen~\cite{ChanCh07}
to produce the edges of $\RNG(S)$ in sorted order of length.
Note that having edges of $\RNG(S)$ in sorted order is necessary only in 
the algorithm in Section~\ref{sec:s-net-emst}, 
where we introduce the $s$-net structure.
More precisely, in order to update the $s$-net efficiently, we must
add the edges of $\RNG(S)$ one by one in their sorted order.
Nevertheless, this procedure is also exploited in our simple algorithm 
in Section~\ref{sec:simple-emst} with the aim of reporting edges of 
$\EMST(S)$ in the sorted order of their length instead of in an arbitrary order.

\begin{lemma}\label{lem:gen_edges}
	Let $S$ be a planar set of $n$ point-sites in general position stored in
	a read-only array. 
	Let $s \in \{1, \dots, n\}$ be a parameter.
	Let $E_R = e_1, e_2, \dots, e_m$ be the sequence of edges 
	in $\RNG(S)$ sorted by increasing length. Let $i \geq 1$.
	Given $e_{i-1}$ (or $\perp$, if $i = 1$), we can find the next $s$ edges 
	$e_{i}, \dots, e_{i + s - 1}$ in $E_R$ using $O\big( (n^2/s) \log s \big)$ 
	time and $O(s)$ cells of workspace.\footnote{Naturally, if $i + s - 1 > m$,
		we report the edges $e_i, \dots, e_m$.}
\end{lemma}

\begin{proof}
	The algorithm in Theorem~\ref{thm:rng} generates all edges of $\RNG(S)$
	in $O \big( (n^2/s) \log s \big)$ time. As we have seen, each step of 
	this algorithm produces a batch of $O(s)$ edges of $\RNG(S)$, using
	Lemma~\ref{lem:6neighbors}. Now after each step of this algorithm,
	instead of reporting the edges, we select the edges 
	$e_{i}, \dots, e_{i + s - 1}$ among them, and we store these edges 
	in the workspace.
	This can be done with a trick by Chan and Chen~\cite{ChanCh07}:
	when the algorithm produces $O(s)$ new edges of $\RNG(S)$, 
	we store the edges that are longer than $e_{i-1}$ in an array $A$ 
	of size $O(s)$. Whenever $A$ contains more than $2s$ elements, 
	we use a linear time selection
	procedure to remove all the edges of rank larger than 
	$s$~\cite{CormenLeRiSt09}. 
	This needs $O(s)$ operations for each batch in the algorithm 
	of Theorem~\ref{thm:rng},
	giving a total of $O(n)$ time for selecting the edges. In the 
	end, we have 
	$e_{i}, \dots, e_{i + s - 1}$ in $A$, albeit not in sorted order. 
	Thus, we sort the final $A$ in $O(s\log s)$ time. 
	The running time for selecting the edges and sorting them is 
	dominated by the time needed to compute all the edges of $\RNG(S)$.
	The space usage for generating the edges and also for selecting 
	and sorting
	them is bounded by $O(s)$ cells of workspace. Thus, the claim follows.
\end{proof}

\section{A Simple Time-Space Trade-Off for EMST}
\label{sec:simple-emst}

The algorithm in Theorem~\ref{thm:rng} for producing edges of $\RNG(S)$,
together with the techniques from the constant-workspace algorithm 
by Asano~\etal~\cite{AsanoMuRoWa11} described in Section~\ref{sec:def-emst}, 
leads to a simple time-space trade-off for computing $\EMST(S)$ 
that we will explain now.

\subsection{Structure of Face-Cycles}
Recall from Section~\ref{sec:def-rng} 
that a partial relative neighborhood graph $\RNG_i$ 
is represented as a collection of face-cycles. 
As described in Section~\ref{sec:def-emst}, 
Asano~\etal~\cite{AsanoMuRoWa11}  have observed that,
to run Kruskal's algorithm on $\RNG(S)$,
it suffices to know the structure of
the face-cycles of $\RNG_i$, for $i \in \{1, \dots, m\}$.
The following observation makes this precise.

\begin{observation}\label{obs:face_cycle}
	Let $i \in \{1, \dots, m\}$.
	The edge $e_i \in E_R$ does not belong to $\EMST(S)$ if and only
	if there is a face-cycle $F$ in $\RNG_i$ such that both endpoints
	of $e_i$ lie on $F$.
\end{observation}

\begin{proof}
	Let $p$ and $q$ be the endpoints of $e_i$. If there is a face-cycle $F$ in $\RNG_i$ that 
	contains both $p$ and $q$, then $e_i$ clearly does not belong to $\EMST(S)$;
	see Figure~\ref{fig:endpoints-face-cycles-a}.
	Conversely, suppose that $e_i$ does not belong to $\EMST(S)$ and 
    hence $p$ and $q$ lie in the same 
    component of $\RNG_i$. Since $e_i$ does not belong to $\RNG_i$, 
    and since $\RNG(S)$ is plane, there is a face $\Gamma$ of 
    $\RNG_i$ such that $e_i \subset \Gamma$. Thus, $p$ and $q$ lie 
    on the boundary $\partial \Gamma$ of $\Gamma$ and in fact, since $p$ and $q$
    are
    in the same component of $\RNG_i$, they lie in the same 
    component  $F$ of $\partial \Gamma$. Then, $F$ is a face-cycle 
    that contains both $p$ and $q$;
	see Figure~\ref{fig:endpoints-face-cycles-b}.
\end{proof}

\begin{figure}[ht]
	\centering
	\subcaptionbox{\label{fig:endpoints-face-cycles-a}}{\includegraphics[page=1]{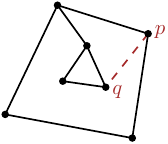}}
	\hspace{1.5cm}
	\subcaptionbox{\label{fig:endpoints-face-cycles-b}}{\includegraphics[page=2]{Endpoints-Face-Cycles}}
	\caption{A schematic drawing of $\RNG_i$. (a) The edge $pq \not \in \EMST(S)$ since
		there is a face-cycle that both $p$ and $q$ lie on.
		(b) If $p$ and $q$ were in the same connected component 
		of $\RNG_i$, but if there were no face-cycle that 
		contains both of them,
		then $e_i$ would cross an edge of $\RNG_i$, 
		contradicting the planarity of $\RNG(S)$.}
	\label{fig:endpoints-face-cycles}
\end{figure}

Observation~\ref{obs:face_cycle} tells us that we can 
identify edges of $\EMST(S)$ if we can determine for each 
$e_i$ the face-cycles in $\RNG_i$ 
that contain the endpoints of $e_i$,
for $i \in \{1, \dots, m\}$.
To accomplish this task, we use the next lemma
to traverse the face-cycles.

\begin{lemma}\label{lem:face_traversal-next}
	Let $i, j \in \{1, \dots, m\}$ and $i > j$.
	Suppose we are given the length $|e_i|$ of $e_i \in E_R$,  a half-edge
	$\overrightarrow{e_j}$ of $e_j \in E_R$ and the
	edges incident to the head of $\overrightarrow{e_j}$ in $\RNG(S)$
	(there are at most six such edges).
	Let $F$ be the face-cycle of $\RNG_i$ that $\overrightarrow{e_j}$ lies on. 
	We can find the next half-edge of $\overrightarrow{e_j}$
	on $F$ in $O(1)$ time using $O(1)$ cells of workspace.
\end{lemma}

\begin{proof}
	Let $\overrightarrow{f_j}$ be the next half-edge of 
    $\overrightarrow{e_j}$ on $F$. Let $q$ be the head of $\overrightarrow{e_j}$. 
	By comparing the length of the edges incident to $q$ in $\RNG(S)$ with $|e_{i}|$, 
	we identify the ones that appear in
	$\RNG_{i}$, in $O(1)$ time. Then, among them we pick the
	half-edge $\overrightarrow{f_j}$ which has the 
	smallest clockwise angle with $\overrightarrow{e_j}$ 
	around $q$ and has $q$ as its tail. This takes $O(1)$ time
	using $O(1)$ cells of workspace; see Figure~\ref{fig:next-edge}.
\end{proof}

\begin{figure}[ht]
	\centering
	\includegraphics[page=1]{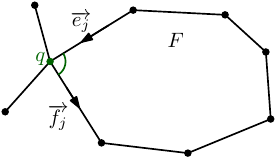}
	\caption{For $i>j$, a schematic drawing of a face-cycle $F$ of $\RNG_i$, and
		$\protect\overrightarrow{e_j}$ on $F$ with the head
		vertex $q$, as well as the other edges of $\RNG_i$ incident to $q$.
		The edge $\protect\overrightarrow{f_j}$
		which has the smallest clockwise angle with $\protect\overrightarrow{e_j}$
		is the next edge of $\protect\overrightarrow{e_j}$ on $F$.}
	\label{fig:next-edge}
\end{figure}

\begin{lemma}\label{lem:face_traversal-pre-suc}
	Let $i, j \in \{1, \dots, m\}$ and $i \leq j$.
    Suppose we are given the length $|e_i|$ of $e_i \in E_R$,  a 
    half-edge $\overrightarrow{e_j}$ of $e_j \in E_R$ and 
	the
	edges incident to the head of $\overrightarrow{e_j}$ in $\RNG(S)$
	(there are at most six such edges).
	We can find $\text{pre}(\overrightarrow{e_j})$ 
    and $\text{suc}(\overrightarrow{e_j})$ in $\RNG_i$
	in $O(1)$ time using $O(1)$ cells of workspace.
\end{lemma}

\begin{proof}
	Let $q$ be the head of $\overrightarrow{e_j}$.
	By comparing the length of the edges incident to $q$ in $\RNG(S)$
	with $|e_{i}|$, we identify the incident half-edges of $q$ in
	$\RNG_{i}$ in $O(1)$ time. Then, among them we pick the
	half-edge $\text{pre}(\overrightarrow{e_j})$ which has $q$ as its head and
	makes the smallest counterclockwise angle with $\overrightarrow{e_j}$ 
	around $w$. Similarly, we pick the half-edge $\text{suc}(\overrightarrow{e_j})$
	which has $q$ as its tail and makes the smallest clockwise angle with $\overrightarrow{e_j}$. 
	This takes $O(1)$ time
	using $O(1)$ cells of workspace; see Figure~\ref{fig:pre-suc-edge}.
\end{proof}

\begin{figure}[ht]
	\centering
	\includegraphics[page=2]{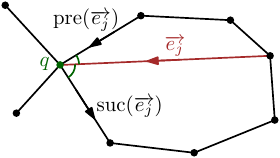}
	\caption{For $i \leq j$, a schematic drawing of $\RNG_i$ (in black) 
		and a half-edge $\protect\overrightarrow{e_j}$ with head $q$. 
		The half-edge $\text{suc}(\protect\overrightarrow{e_j})$ 
        has the smallest clockwise angle with $\protect\overrightarrow{e_j}$.
      		The half-edge $\text{pre}(\protect\overrightarrow{e_j})$ 
        has the smallest counterclockwise angle with $\protect\overrightarrow{e_j}$.}
	\label{fig:pre-suc-edge}
\end{figure}

\subsection{The Algorithm}
From our observations so far, we can
derive a simple time-space trade-off for computing
$\EMST(S)$. In Theorem~\ref{thm:simple_algo},
we simulate Kruskal's algorithm on $\RNG(S)$.
For this, we take batches of $s$ edges of $\RNG(S)$,
and we report the edges of $\EMST(S)$ among them. 
To determine whether an edge $e_i$ of $\RNG(S)$ is in $\EMST(S)$,
we apply Observation~\ref{obs:face_cycle}, i.e., 
we determine whether the endpoints of $e_i$ are
on a common face-cycle in the corresponding $\RNG_i$.

\begin{theorem}\label{thm:simple_algo}
	Let $S$ be a planar set of $n$ point-sites in general position
	stored in a read-only array. Let $s \in \{1, \dots, n\}$
	be a parameter. We can output all the edges of $\EMST(S)$,
	in sorted order of their length, in $O\big( (n^3/s)\log s \big)$ 
	time using $O(s)$ cells of workspace.
\end{theorem}

\begin{proof}
	Let $E_R = e_1, \dots, e_m$ be the sequence of edges of $\RNG(S)$, sorted by length. 
	In the first iteration, we use Lemma~\ref{lem:gen_edges} to find the batch 
	$e_1, \dots, e_s$ of the first $s$ edges in $E_R$ in 
	$O\big( (n^2/s)\log s \big)$ time.
	For each edge $e_i$, $i\in \{1, \dots, s\}$, we consider 
	both its half-edges. Then, we perform $2s$ parallel walks starting from 
	the head vertex of each half-edge $\overrightarrow{e_i}$.
	In the first step of the walks,  using Lemma~\ref{lem:6neighbors},
	we find the incident edges to the head of each half-edge 
	$\overrightarrow{e_i}$ (there are at most six such edges).
	Then, using Lemma~\ref{lem:face_traversal-pre-suc}, we identify 
	$\text{pre}(\overrightarrow{e_i})$ and $\text{suc}(\overrightarrow{e_i})$
	in $\RNG_i$ (if they exist). 
	By following the successor of each half-edge, we perform
	one step of the walk for each half-edge of the batch in parallel.
	Note that the walk that starts from the head of $\overrightarrow{e_i}$ takes
	place in $\RNG_i$.
	
	Next, in the second step of the parallel walks, we consider the head vertices
	of all the $\text{suc}(\overrightarrow{e_i})$.
	First, we use Lemma~\ref{lem:6neighbors}
	to find the incident edges to the head of each 
    $\text{suc}(\overrightarrow{e_i})$ (there are at most six such edges).
	Then, applying Lemma~\ref{lem:face_traversal-next}, we find the next half-edge of
    $\text{suc}(\overrightarrow{e_i})$,
	and we advance each half-edge along 
	its face-cycle in $\RNG_i$ as one step of the parallel walks.
	We proceed the parallel walks by finding the next edge on the face-cycles in
	each step. 
	
	A walk that started from the head $q$ of $\overrightarrow{e_i}$ continues until 
	it either encounters the tail $p$ of  $\overrightarrow{e_i}$ or
	until it arrives at $\text{pre}(\overrightarrow{e_i})$.
	In the former case, we have found a face-cycle that both
	endpoints of $e_i$ lie on and thus, by Observation~\ref{obs:face_cycle}, 
	$e_i$ is not in $\EMST(S)$; see Figure~\ref{fig:parallel-walks-a}. 
	In the latter case, there is no face-cycle in $\RNG_i$ that contains both $p$ 
	and $q$. This is because, by definition of $\text{pre}(\overrightarrow{e_i})$
	and $\text{suc}(\overrightarrow{e_i})$, all the incident edges of $q$ in $\RNG_i$ lie in
	the counterclockwise cone between 
    $\text{pre}(\overrightarrow{e_i})$ and $\text{suc}(\overrightarrow{e_i})$
	around $q$. Therefore, by planarity of $\RNG_i$,
	all the other face-cycles that contain $q$ are separated from $p$ 
    by the face-cycle that
	starts with $\text{suc}(\overrightarrow{e_i})$ and ends at 
    $\text{pre}(\overrightarrow{e_i})$.
	Hence, none of those face-cycles encounters $p$ and, 
    by Observation~\ref{obs:face_cycle}, $e_i$ is an edge of 
    $\EMST(S)$; see Figure~\ref{fig:parallel-walks-b}.
	In this case, we report 
	$e_i$, and we also abort the walk that was started from the opposite half-edge 
	of $\overrightarrow{e_i}$. This prevents an edge of $\EMST(S)$ to be reported twice.
	
	\begin{figure}[ht]
		\centering
		\subcaptionbox{\label{fig:parallel-walks-a}}{\includegraphics[page=1]{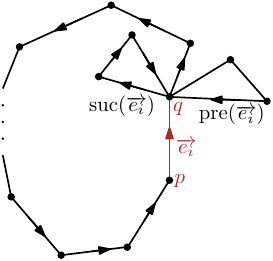}}
		\hspace{1.5cm}
		\subcaptionbox{\label{fig:parallel-walks-b}}{\includegraphics[page=2]{Parallel-Walks}}
		\caption{A schematic drawing of $\RNG_i$ and the half-edge $\protect\overrightarrow{e_i}$
			with head $q$ and tail $p$. (a) The vertices $p$ and $q$ are on the same 
			face-cycle of $\RNG_i$ since by traversing the face-cycle starting from
			$\text{suc}(\protect\overrightarrow{e_i})$ we encounter $p$.
			(b) The vertices $p$ and $q$ are on different face-cycles of $\RNG_i$ since
			by traversing the face-cycle starting from
            $\text{suc}(\protect\overrightarrow{e_i})$ we encounter
			$\text{pre}(\protect\overrightarrow{e_i})$, meaning that we will not reach $p$.
		}
		\label{fig:parallel-walks}
	\end{figure}
	
	In the next iteration of the algorithm, we again use 
	Lemma~\ref{lem:gen_edges} to find the next batch of $s$ edges 
	in $E_R$. Similarly, we perform $2s$ parallel walks for the half-edges in 
	this batch, in order to find the edges that belong to $\EMST(S)$.
	
	Since there are $O(n)$ half-edges in $\RNG(S)$, it takes $O(n)$ steps
	in each iteration to conclude all the walks, where each step of the walks takes
	$O(n\log s)$ time. It follows that we can process a single batch of edges in $O(n^2\log s)$ 
	time which dominates the time needed for finding a batch of $s$ edges of $\RNG(S)$.
	We have $O(n/s)$ batches, so the total running time of the 
	algorithm is $O\big( (n^3/s)\log s \big)$.
	The algorithm uses $O(s)$ cells of workspace for finding and storing 
	a batch of $s$ edges as well as a constant number of cells 
	per edge to perform each walk.
\end{proof}

Note that, in this algorithm, it is not essential to process edges 
of $\RNG(S)$ in sorted order of length. Thus, we can simply apply 
Lemma~\ref{lem:6neighbors} to produce edges of $\RNG(S)$.
However, by using Lemma~\ref{lem:gen_edges} we are able to
report edges of $\EMST(S)$ in sorted order of length,
although the total running time of the algorithm will not be affected.

\section{Improvement via a Compact Representation of RNGs}
\label{sec:s-net-emst}

Theorem~\ref{thm:simple_algo} is clearly not optimal:
for the case of linear space $s = n$, we get a running time 
of $O(n^2 \log n)$, 
although we know that it is possible to find $\EMST(S)$ 
in $O(n \log n)$ time.
Can we do better?
The bottleneck in Theorem~\ref{thm:simple_algo} is the time needed to
perform the walks in the partial relative neighborhood
graphs $\RNG_i$. In particular, such a walk might take 
$\Omega(n)$ steps, leading to a running time of
$\Omega(n^2 \log s)$ for processing a single batch of $s$ edges.
To avoid this, we will maintain a compressed representation
of the partial relative neighborhood graphs that allows
us to reduce the number of steps in each walk to $O(n/s)$.

\subsection{The  $s$-net Structure}
Let $i \in \{1, \dots, m\}$. 
An \emph{$s$-net} $N$ for $\RNG_i$ is a collection of half-edges, 
called \emph{net-edges}, in 
$\RNG_i$ that has the following two properties:
(i) Each face-cycle in $\RNG_i$ with at least $\lfloor n/s \rfloor + 1$
half-edges contains at least one net-edge. (ii) For any
net-edge $\overrightarrow{e} \in N$, let $F$ be the face-cycle of 
$\RNG_i$ that contains $\overrightarrow{e}$. Then on $F$, between the head of 
$\overrightarrow{e}$ and the tail of the next net-edge, there are
at least $\lfloor n/s \rfloor$ and at most 
$2 \lfloor n/s \rfloor$ other half-edges.
Note that the next net-edge on $F$ after $\overrightarrow{e}$
could possibly be $\overrightarrow{e}$ itself.
In particular, this implies that face-cycles with less than 
$\lfloor n/s \rfloor$ edges contain no net-edges;
see Figure~\ref{fig:s-net}.

\begin{figure}[ht]
	\centering
	\subcaptionbox{\label{fig:s-net-a}}{\includegraphics[page=1]{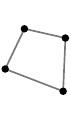}}
	\hspace{1cm}
	\subcaptionbox{\label{fig:s-net-b}}{\includegraphics[page=2]{s-net}}
	\hspace{1cm}
	\subcaptionbox{\label{fig:s-net-c}}{\includegraphics[page=3]{s-net}}
	\caption{A schematic drawing of an $s$-net for $\RNG_i$.
		(a) A small face-cycle with no net-edges. (b) a face-cycle with
		more than $\lfloor n/s \rfloor$ and less than $2 \lfloor n/s \rfloor$
		half-edges which contains one net-edge.
		(c) A big face-cycle with four net-edges.}
	\label{fig:s-net}
\end{figure}

We note two important observations about $s$-nets.

\begin{observation}\label{obs:s-net}
	Let $i \in \{1, \dots, m\}$, and let $N$ be an $s$-net
	for $\RNG_i$. Then,
	\begin{enumerate}[({N}1)]
		\item
		$N$ has $O(s)$ half-edges; and
		\item
		let $\overrightarrow{f}$ be a half-edge of $\RNG_i$, and
		let $F$ be the face-cycle that contains it. Then, it takes
		at most $2\lfloor n/s \rfloor$ steps along $F$ from the head 
		of $\overrightarrow{f}$ until we encounter the tail of 
		either a net-edge or $\overrightarrow{f}$ itself.
	\end{enumerate}
\end{observation}

\begin{proof}
	Property~(ii) of the definition of an $s$-net
	implies that only face-cycles of $\RNG_i$ with
	at least $\lfloor n/s \rfloor+1$ half-edges contain net-edges. Furthermore, on these
	face-cycles, we can uniquely charge $\Theta(n/s)$ half-edges to each
	net-edge, again by property~(ii). Since
	the face-cycles of $\RNG_i$ have $O(n)$ half-edges in total, 
	we obtain the first observation which says $|N| = O(s)$.
	
	For the second observation, we first note that if $F$ contains less
	than $2 \lfloor n/s \rfloor$ half-edges, the claim holds trivially. 
	Otherwise, by property~(i), $F$ contains at least one net-edge. 
	From property~(ii) it follows that there are at most $2 \lfloor n/s \rfloor$
	half-edges between every two consecutive net-edges on $F$. Thus, in a walk 
	on $F$ starting from $\overrightarrow{f}$, we reach a net-edge
	in at most $2 \lfloor n/s \rfloor$ steps.
\end{proof}

Due to statement~(\emph{N1}) of Observation~\ref{obs:s-net}, an
$s$-net can be stored in $O(s)$ cells of workspace. This makes the
concept of $s$-net useful in our algorithm with a workspace of $O(s)$ cells.
Therefore, we can exploit the $s$-net in order to speed up the
processing of a single batch. The next lemma shows how
this is done.

\begin{lemma}\label{lem:net_batch}
	Let $i \in \{1, \dots, m\}$. Suppose we are given 
	$E_{i,s}= e_{i}, \dots, e_{i + s - 1}$, a batch of $s$ edges in $E_R$.
	Furthermore, we have an $s$-net $N$ for $\RNG_i$ in our 
	workspace. Then, we can determine which edges from
	$E_{i,s}$ belong to $\EMST(S)$ in
	$O\big( (n^2/s)\log s\big)$ time using $O(s)$ cells of workspace.
\end{lemma}

\begin{proof}
	Let $H$ be a set of half-edges defined as 
	follows: the set $H$ is the union of 
	all net-edges from $N$, and, for each \emph{batch-edge}
	$e_j \in E_{i,s}$, the successors of the two half-edges of $e_j$
	in $\RNG_i$; see Figure~\ref{fig:set-H}.
	
	\begin{figure}[ht]
		\centering
		\includegraphics[page=1]{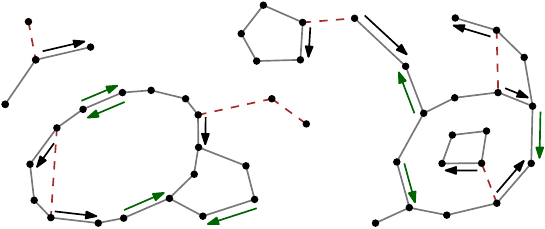}
		\caption{A schematic drawing of $\RNG_i$ with a batch of edges 
			in $E_R$ (dashed red segments). The directed segments
			represents the half-edges in $H$. The net-edges are in green and the successors of
			the batch edges are in black.}
		\label{fig:set-H}
	\end{figure}
	
	By definition, we have $|H| = O(s)$, and thus it takes $O(n \log s)$ time
	to compute $H$. This is done by using Lemma~\ref{lem:6neighbors}
	to find the incident edges of the head of each $e_j$ and
	Lemma~\ref{lem:face_traversal-pre-suc} to identify the successors
	of each $e_j$.
	
	Now starting from the half-edges in $H$, we perform 
	parallel walks through the face-cycles of $\RNG_i$, 
	one walk per half-edge. Each such walk
	proceeds until it encounters the tail of a half-edge in $H$ 
	(including the starting half-edge itself).
	In each step of these walks, we use
	Lemma~\ref{lem:6neighbors} and Lemma~\ref{lem:face_traversal-next} 
	to find the next half-edges on the face-cycles in $O(n\log s)$ time,
	and then we check whether these new half-edges belong to $H$ in $O(s\log s)$ time.
	Because $H$ contains the net-edges of $N$, by statement~(\emph{N2}) of
	Observation~\ref{obs:s-net}, 
	each walk finishes after $O(n/s)$ steps, and thus,
	the total time for this procedure is
	$O \big( (n^2/s) \log s \big)$.
	
	Next, we build an auxiliary \emph{undirected} (multi-)graph $G$ as follows:
	the vertices of $G$ are the endpoints of the half-edges in
	$H$ and the endpoints of the half-edges of $E_{i,s}$.\footnote{Not 
	all the endpoints of half-edges in $E_{i,s}$ are 
	necessarily included as endpoints of half-edges in $H$: 
	the successor of a half-edge from $E_{i,s}$ might be
	\emph{Null}. In this case, we still want to include the
	endpoints of this half-edge in $G$.}
	Furthermore, $G$ contains 
	undirected edges for all the
	half-edges in $H$
	and additional \emph{compressed edges}, that represent the 
	outcomes of the walks: if a walk
	started from the head $q$ of a half-edge in $H$ and ended at the
	tail $p$ of a half-edge in $H$, we add an edge from $q$ to $p$ in
	$G$, and we label it with the number of steps that were needed
	for the walk, i.e., the number of half-edges between $q$ and $p$ on that face-cycle.
	Thus, $G$ contains \emph{$H$-edges},
	and \emph{compressed edges}; 
	see Figure~\ref{fig:emst-graph-G}.
	Clearly, after all the walks have been terminated,
	we can construct $G$ in $O(s)$ time, using $O(s)$ cells of workspace. 
	
	\begin{figure}[ht]
		\centering
		\includegraphics[page=2]{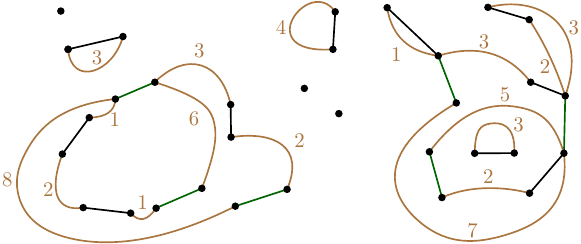}
		\caption{The auxiliary graph $G$ is shown. The edge set of $G$ contains
			the net-edges (in green), the successors of
			batch-edges (in black), and the compressed edges (beige paths).}
		\label{fig:emst-graph-G}
	\end{figure}
	
	The auxiliary graph $G$ is actually a representation of the face-cycles
	in $\RNG_i$. Thus, by adding the batch-edges of $E_{i,s}$ 
	one by one into $G$, we can represent the next partial relative neighborhood
	graphs, up to $\RNG_{i+s}$. Hence, we can use $G$ to identify which of the 
	batch-edges of $E_{i,s}$ belong to $\EMST(S)$. 
	This is done by applying Kruskal's algorithm on $G$ as follows:
	we determine the connected components of $G$ in $O(s)$ time using
	depth-first search. Then, we insert the batch-edges into $G$, 
	one after another, in sorted order. As we do this,
	we keep track of how the connected components of $G$ change, using
	a union-find data structure~\cite{CormenLeRiSt09}. Whenever
	a new batch-edge connects two distinct connected components of $G$,
	we output it as an edge of $\EMST(S)$. Otherwise, we do nothing; 
	see Figure~\ref{fig:batch-edges-graph-G}.
	Note that even though one component of $\RNG_i$ might
	be represented by several components in $G$,\footnote{Two 
	(or several) face-cycles in one 
	component of $\RNG_i$ may share some vertices. However, these 
	vertices need not necessarily appear as vertices in $G$. Hence,
	representing those face-cycles with compressed edges, one might not 
	represents their common parts in $G$. 
	Therefore, such face-cycles might belong to distinct components 
	in $G$.},
	the algorithm is still correct because of 
	Observation~\ref{obs:face_cycle}.
	
	This execution of Kruskal's algorithm and updating the structure of
	connected components of $G$ takes $O(s \log s)$ time, which
	is dominated by the running time of $O\big((n^2/s)\log s\big)$
	to perform the parallel walks. The space requirement for constructing
	and storing the set $H$ and the graph $G$ as well as the updated
	versions of $G$ is a total of $O(s)$ cells of workspace.
\end{proof}

\begin{figure}[ht]
	\centering
	\includegraphics[page=3]{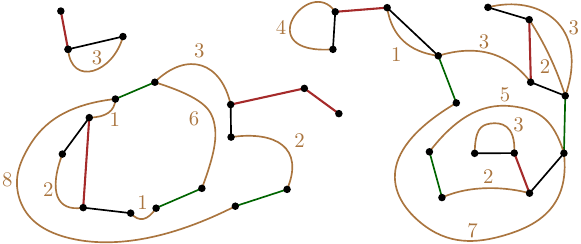}
	\caption{The batch-edges of $E_{i,s}$ (in red) have been added to the auxiliary graph $G$.}
	\label{fig:batch-edges-graph-G}
\end{figure}

\subsection{Maintaining the $s$-net}
Now that we have described how to use an $s$-net for $\RNG_i$
in order to process the edges $e_i, \dots, e_{i+s}$ of $E_R$, 
we need to explain how to maintain the $s$-net during the
algorithm, i.e., how to construct an $s$-net for $\RNG_{i+s}$
after processing the edges $e_i, \dots, e_{i+s}$.
The algorithm in the following lemma computes an $s$-net for
$\RNG_{i+s}$, provided that we have an $s$-net for $\RNG_i$ as well as
the graph $G$ as it is constructed in the proof of Lemma~\ref{lem:net_batch},
for each $i \in \{1, \dots, m\}$.

\begin{lemma}\label{lem:new_net}
	Let $i \in \{1, \dots, m\}$, and suppose we have the graph
	$G$ derived from $\RNG_i$ as above, such that all batch-edges have been 
	inserted into $G$. Then, we can compute an $s$-net $N$ for
	$\RNG_{i+s}$ in time $O\big( (n^2/s) \log s\big)$, using $O(s)$ cells of workspace.
\end{lemma}

\begin{proof}
	By construction, all \emph{big} face-cycles of $\RNG_{i+s}$, i.e., those 
	face-cycles with at least $\lfloor n/s \rfloor + 1$ half-edges, 
	appear as faces in $G$. Thus, by walking along 
	all faces in $G$, and taking into account the labels of the compressed
	edges, we can determine these big face-cycles in $O(s)$ time.
	The big face-cycles are represented through sequences of $H$-edges, compressed
	edges, and batch-edges. For each such sequence, we determine the positions
	of the half-edges for the new $s$-net $N$, by spreading the half-edges 
	equally at minimum distance $\lfloor n/s \rfloor$ and maximum distance
	$2 \lfloor n/s \rfloor$ along the 
	sequence, again taking the labels of the compressed edges into
	account. Since the compressed edges have length $O(n/s)$, for each of them,
	we create at most $O(1)$ new net-edges. Now that we have determined 
	the positions of the new net-edges on the face-cycles of $\RNG_{i+s}$,
	we perform $O(s)$ parallel walks in $\RNG_{i+s}$ to actually find them.
	Using Lemma~\ref{lem:6neighbors} and Lemma~\ref{lem:face_traversal-next}, 
	this takes $O\big( (n^2/s)  \log s \big)$ time; see Figure~\ref{fig:new-net-edges}.
\end{proof}

\begin{figure}[ht]
	\centering
	\subcaptionbox{\label{fig:new-net-edges-a}}{\includegraphics[page=1]{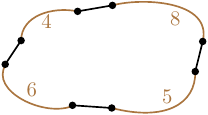}}
	\hspace{1cm}
	\subcaptionbox{\label{fig:new-net-edges-b}}{\includegraphics[page=2]{New-Net-Edges}}
	\caption{(a) A schematic drawing of a face-cycle in $G$ and (b) distributing the
		new net-edges (in green) on this face-cycle with almost equal distances.}
	\label{fig:new-net-edges}
\end{figure}

We now have all the ingredients for our main result that provides a smooth 
trade-off between the cubic-time algorithm in constant workspace and the classic
$O(n\log n)$-time algorithm with $O(n)$ cells of workspace. The following
theorem presents this algorithm.

\begin{theorem}\label{thm:main_algo}
	Let $S$ be a planar set of $n$ point-sites
	in general position stored in a read-only array. Let $s \in \{1, \dots, n\}$
	be a parameter. We can report all the edges of $\EMST(S)$,
	in sorted order of length, in $O\big( (n^3/s^2) \log s \big)$ 
	time using $O(s)$ cells of workspace.
\end{theorem}

\begin{proof}
	This follows immediately from our lemmas: 
	applying Lemma~\ref{lem:gen_edges}, we
	produce a batch of $s$ edges of $\RNG(S)$ in sorted order of length.
	Then, among them, we report the edges of $\EMST(S)$, using Lemma~\ref{lem:gen_edges}.
	Finally, we maintain the $s$-net structure to be used for the next
	batch of $s$ edges of $\RNG(S)$, by Lemma~\ref{lem:new_net}.
	All these steps are done in $O\big( (n^2/s) \log s \big)$ time using $O(s)$
	cells of workspace. Since $\RNG(S)$ has $O(n)$ edges, we need to process
	$O(n/s)$ batches of edges of $\RNG(S)$, leading to an
	algorithm with total running time of $O\big( (n^3/s^2) \log s \big)$,
	and total workspace usage of $O(s)$ cells.
\end{proof}

\section{Conclusion}

For our algorithm, it suffices to update the $s$-net every 
time that a new batch is considered. It is, however, possible 
to maintain the $s$-net and the auxiliary graph $G$ through 
insertions of single edges, with the same bound as in Lemma~\ref{lem:new_net}. 
This allows us to handle graphs constructed incrementally and 
to maintain their compact representation using $O(s)$ workspace cells.
We believe this is of independent interest and can be used 
by other algorithms for planar graphs in the limited-workspace model.

Also, it remains an intriguing question whether the EMST can be 
computed in $o(n^3)$ time in the constant-workspace model.
Intuitively, it seems hard to improve the $O(n^2)$-time algorithm for 
checking whether an individual edge belongs to the EMST, and maybe 
it will be possible to obtain a formal lower bound for this 
subproblem. However, even such a lower bound would not rule out 
other possible approaches towards a faster EMST-algorithm.

\paragraph{Acknowledgments.}
This work was initiated at the 
Fields Workshop on Discrete and Computational Geometry,
held 
07.31.--08.04.2017, at Carleton university.
The authors would like to thank them and
all the participants of the workshop for inspiring discussions
and for providing a great research atmosphere.

\bibliographystyle{abbrv}
\bibliography{thesis}

\end{document}